\documentclass{cccg15}

\usepackage{epsfig}
\usepackage{graphicx}
\usepackage{color}
\usepackage{times}
\usepackage{amsmath}
\usepackage{amssymb}
\usepackage{xspace,cleveref}
\usepackage{url}

\newcommand {\mm}[1] {\ifmmode{#1}\else{\mbox{\(#1\)}}\fi}

\newcommand{\denselist}{\itemsep 0pt\parsep=1pt\partopsep 0pt}

\newcommand{\ignore}[1]{}

\newsavebox{\smallProofsym}                            
\savebox{\smallProofsym}                               %
{
\begin{picture}(6,6)                                   %
\put(0,0){\framebox(6,6){}}                            %
\put(0,2){\framebox(4,4){}}                            %
\end{picture}                                          %
}                                                      %
\makeatletter
\long\def\@makecaption#1#2{%
  \vskip\abovecaptionskip
  \sbox\@tempboxa{\small #1: #2}%
  \ifdim \wd\@tempboxa >\hsize
    \small #1: #2\par
  \else
    \global \@minipagefalse
    \hb@xt@\hsize{\hfil\box\@tempboxa\hfil}%
  \fi
  \vskip\belowcaptionskip}
\makeatother

\newcommand{\Rspace}        {\mm{{\mathbb R}}}
\newcommand{\Zspace}        {\mm{{\mathbb Z}}}

\newcommand{\norm}[1]       {\mm{\|{#1}\|}}
\newcommand{\Edist}[2]      {\mm{\|{#1}-{#2}\|}}

\title{Relaxed Disk Packing\thanks{This work
  is partially supported by the {\sc Toposys} project FP7-ICT-318493-STREP,
  and by ESF under the ACAT Research Network Programme.}}

\author{Herbert Edelsbrunner\thanks{IST Austria (Institute of Science
          and Technology Austria), Kloster\-neu\-burg, Austria,
          \texttt{edels@ist.ac.at}.} \and
        Mabel Iglesias-Ham\thanks{IST Austria (Institute of Science
          and Technology Austria), Kloster\-neu\-burg, Austria,
          \texttt{mabel.iglesias-ham@ist.ac.at}.} \and
        Vitaliy Kurlin\thanks{Microsoft Research Cambridge
          and Mathematical Sciences, Durham University,
          United Kingdom, \texttt{vitaliy.kurlin@gmail.com}.}
}

\begin{document}
\maketitle

\begin{abstract}
  Motivated by biological questions, we study configurations of equal-sized
  disks in the Euclidean plane that neither pack nor cover.
  Measuring the quality by the probability that a random point lies
  in exactly one disk, we show that the regular hexagonal grid
  gives the maximum among lattice configurations.
\end{abstract}

\vspace{0.1in}
{\small
 \noindent{\bf Keywords.}
   Packing and covering, disks, lattices, Voronoi domains,
   Delaunay triangulations.}

\section{Introduction}
\label{sec1}

High-resolution microscopic observations of the DNA organization inside
the nucleus of a human cell support
the \emph{Spherical Mega-base-pairs Chromatin Domain model}
\cite{Crem00,Kreth01}.
It proposes that inside the chromosome territories in eukaryotic cells,
DNA is compartmentalized in sequences of highly interacting segments
of about the same length \cite{Dixon12}.
Each segment consists of roughly a million base pairs and
resembles a round ball.  
The balls are tightly arranged within a restricted space,
tighter than a packing since they are not rigid,
and less tight than a covering to allow for external access
to the DNA needed for gene expression.

Motivated by these biological findings, \cite{Ham14} considered
configurations in which the overlap between the balls is limited and the
quality is measured by the density, which we define as the expected
number of balls containing a random point. 
We introduce a new measure that favors configurations between packing
and covering without explicit constraints on the allowed overlap.
Specifically, we measure a configuration by
\emph{the probability a random point is contained in exactly one ball}.
Since empty space and overlap between disks are both discouraged,
the optimum lies necessarily between packing and covering.
The interested reader can find references for traditional packing
and covering in two and higher dimensions
in \cite{Conway99,FejesToth53}.
In this paper, we restrict attention to equal-sized disks in the plane
whose centers form a lattice,
leaving three and higher dimensions as well as non-lattice configurations
as open problems.
Our main result is the following non-surprising fact.
\begin{theorem}[Main]
  Among all lattice configurations in $\Rspace^2$,
  the regular hexagonal grid in which each disk overlaps the six
  neighboring circles in $30^\circ$ arcs maximizes
  the probability that a random point lies exactly in one disk.
\end{theorem}
For obvious reasons, we call this the \emph{$12$-hour clock configuration}.
We prove its optimality in four sections:
preparing the background in Section \ref{sec2},
proving an equilibrium condition in Section \ref{sec3},
developing the main argument in Section \ref{sec4},
and giving the technical details in Appendix \ref{appA}.

\section{Background}
\label{sec2}

In this section, we introduce notation for lattices, Voronoi domains,
and Delaunay triangulations.

\paragraph{Lattices.}

Depending on the context, we interpret an element of $\Rspace^2$
as a point or a vector in the plane.
Vectors $a, b \in \Rspace^2$ are \emph{linearly independent}
if $\alpha a + \beta b = 0$ implies $\alpha = \beta = 0$.
A \emph{lattice} is defined by two linearly independent vectors,
$a, b \in \Rspace^2$, and consists of all integer combinations of
these vectors:
\begin{align}
  L(a,b)  &=  \{ ia + jb \mid i, j \in \Zspace \} .
\end{align}
Its \emph{fundamental domain} is the parallelogram of points
$\alpha a + \beta b$ with real numbers $0 \leq \alpha, \beta \leq 1$.
Writing $\norm{a}$ for the length of the vector
and $\gamma$ for the angle between $a$ and $b$,
the area of the fundamental domain is
$\det L(a,b) = \norm{a} \norm{b} \sin(\gamma)$.
The same lattice is generated by different pairs of vectors,
and we will see shortly that at least one of these pairs
defines a non-obtuse triangle.
We will be more specific about this condition shortly,
as it is instrumental in our proof of the optimality of
the regular hexagonal grid.

\paragraph{Voronoi domain.}

Given a lattice $L$,
the \emph{Voronoi domain} of a point $p \in L$ is the set of points
for which $p$ is the closest:
\begin{align}
  V(p)  &=  \{ x \in \Rspace^2 \mid \Edist{x}{p} \leq \Edist{x}{q},
                                    \forall q \in L \} .
\end{align}
It is a convex polygon that contains $p$ in its interior.
Any two Voronoi domains have disjoint interiors but may intersect
in a shared edge or a shared vertex.
The lattice looks the same from every one of its points,
which implies that all Voronoi domains are translates of each other:
$V(p) = p + V(0)$.
Similarly, central reflection through the origin preserves the lattice,
which implies that $V(0)$ is centrally symmetric.

\begin{figure}[hbt]
  \centering \resizebox{!}{1.25in}{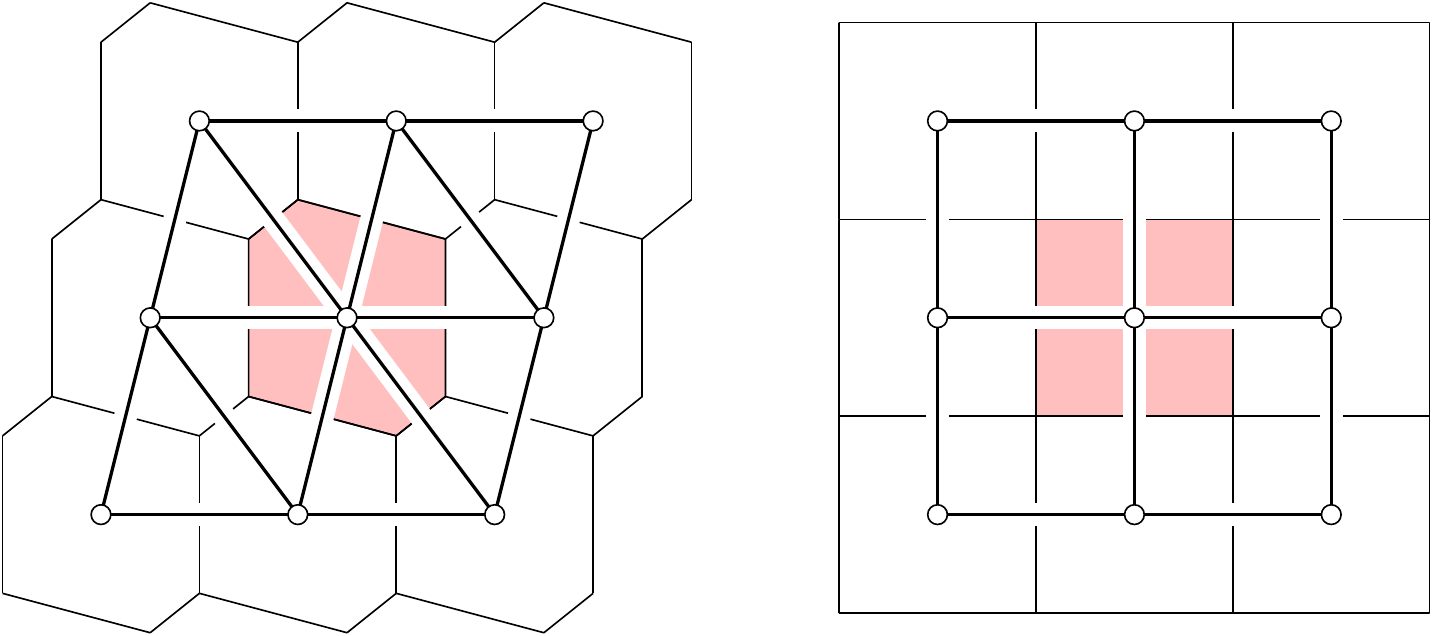}
  \caption{A primitive Voronoi diagram on the \emph{left},
    and a non-primitive Voronoi diagram on the \emph{right}.
    The Delaunay triangulations are superimposed.}
 \label{fig:hexsq}
\end{figure}
The \emph{Voronoi diagram} of $L$ is the collection of Voronoi domains
of its points.
It is \emph{primitive} if the maximum number of Voronoi domains
with non-empty common intersection is $3$.
In this case, the Voronoi domain is a centrally symmetric hexagon;
see Figure \ref{fig:hexsq}.
In the non-primitive case, there are generators that enclose a right angle,
and the Voronoi domains are rectangles.
To the first order of approximation, the area of any sufficiently simple
and sufficiently large subset of $\Rspace^2$ is the number of lattice points
it contains times the area of $V(0)$.
Similarly, it is the number of lattice points
times the area of the fundamental domain.
It follows that the area of $V(0)$ is equal to $\det L$.

\paragraph{Packing and covering.}

For $\varrho > 0$, we write $B(p,\varrho)$ for the closed disk
with center $p$ and radius $\varrho$.
The \emph{packing radius} is the largest radius, $r_L$,
and the \emph{covering radius} is the smallest radius, $R_L$,
such that $B(0,r_L) \subseteq V(0) \subseteq B(0,R_L)$.
The \emph{density} of the configuration of disks with radius $\varrho$
centered at the points of $L$ is the area of a disk divided
by the area of the Voronoi domain:
\begin{align}
  \delta_L (\varrho)  &=  \tfrac{\varrho^2 \pi}{\det L} .
\end{align}
It is also the expected number of disks containing a random point
in $\Rspace^2$.
The \emph{packing density} is $\delta_L (r_L)$,
which is necessarily smaller than $1$.
It is maximized by the regular hexagonal grid, $H$, for which
we have $\delta_H (r_H) = 0.906\ldots$.
The \emph{covering density} is $\delta_L (R_L)$,
which is necessarily larger than $1$.
It is minimized by the regular hexagonal grid for which we have
$\delta_H (R_H) = 1.209\ldots$.
More generally, it is known that $H$ maximizes the density
among all configuration of congruent disks whose interiors
are pairwise disjoint \cite{Thue10},
and it minimizes the density among all configurations
that cover the entire plane \cite{Kershner39}.
Elegant proofs of both optimality results can be found
in Fejes T\'{o}th \cite{FejesToth53}.

\paragraph{Delaunay triangulations.}

Drawing a straight edge between points $p$ and $q$ in $L$
iff $V(p)$ and $V(q)$ intersect along a shared edge, we get
the \emph{Delaunay triangulation} of $L$.
In the primitive case, the edges decompose the plane into triangles.
Among the six triangles sharing $0$ as a vertex,
three are translates of each other and, going around $0$,
they alternate with their central reflections.
It follows that all six triangles are congruent and, in particular,
they have equally large circumcircles that all pass through $0$.
Since their centers are vertices of the Voronoi domain of $0$,
we have the following result.
\begin{lemma}[Inscribed Voronoi Domain]
  The vertices of $V(0)$ all lie on the circle bounding $B(0,R_L)$.
\end{lemma}
The discussion above proves the Inscribed Voronoi Domain Lemma
in the primitive case.
It is also true in the simpler, non-primitive case in which
$V(0)$ is a rectangle.
Returning to the primitive case, we note that the two angles opposite
to a shared edge in the Delaunay triangulation add up to less than $180^\circ$.
In a lattice, these two angles are the same and therefore both acute.
The two types of triangles in the Delaunay triangulation of a lattice
can be joined across a shared edge in three different ways.
We can therefore make the same argument three times and conclude
that all angles are less than $90^\circ$.
A slightly weaker bound holds in the non-primitive case.
\begin{lemma}[Non-obtuse Generators]
  Every lattice $L$ in $\Rspace^2$ has vectors
  $a, b \in \Rspace^2$ with $L = L(a,b)$ such that
  \begin{enumerate}\denselist
    \item[(i)]  in the primitive case $0, a, b$ are the vertices
      of an acute triangle,
    \item[(ii)] in the non-primitive case $0, a, b$ are the vertices
      of a non-obtuse triangle with a right angle at $0$.
  \end{enumerate}
\end{lemma}
Assuming $a, b$ satisfy the Non-obtuse Generators Lemma, the triangle $0ab$
has edges of length $\norm{a}$, $\norm{b}$, and $\norm{c} = \Edist{a}{b}$.
In the non-primitive case (ii),
$\norm{a}$ and $\norm{b}$ are the lengths of the sides of the
rectangle $V(0)$, and $\norm{c}$ is the length of a diagonal.
We have $\norm{c}^2 = \norm{a}^2 + \norm{b}^2$,
and therefore $\norm{a} \leq \norm{b} \leq \norm{c}$,
possibly after swapping $a$ and $b$.
In the primitive case (i), we can choose $a$, $b$, and $c = a-b$
such that $\norm{a} \leq \norm{b} \leq \norm{c}$.
In this case, $V(0)$ is a centrally symmetric hexagon
with distances $\norm{a}, \norm{b}, \norm{c}$ between antipodal
edge pairs.

\section{Equilibrium Configurations}
\label{sec3}

Given a lattice in $\Rspace^2$, we are interested in the radius of
the disks for which the probability that a random point lies inside
exactly one disk is maximized.
Further maximizing this probability over all lattices,
we get the main result of this paper.

\paragraph{Partial disks.}
Fix a lattice $L$ in $\Rspace^2$.
For a radius $\varrho > 0$, consider the set of points
that belong to the disk centered at the origin but not to any other
disk centered at a point of $L$:
\begin{align}
  D(0, \varrho)  &=  B(0, \varrho) \setminus
                     \bigcup_{0 \neq p \in L} B(p, \varrho) .
\end{align}
As illustrated in Figure \ref{fig:partialdisk}, for radii strictly
between the packing radius and the covering radius, this set is
partially closed and partially open.
We distinguish between the \emph{convex boundary} that
belongs to the circle bounding $B(0, \varrho)$,
and the \emph{concave boundary} that belongs to other circles:
\begin{align}
  \partial_x D(0, \varrho)  &=  \partial B(0, \varrho) \cap D(0, \varrho) , \\
  \partial_v D(0, \varrho)  &=  \partial D(0, \varrho) \setminus
                                 \partial_x D(0, \varrho) . 
\end{align}
We note that $\partial_x D(0, \varrho) = \partial B(0, \varrho) \cap V(0)$.
By the Inscribed Voronoi Domain Lemma, the vertices of $V(0)$ are all
at the same distance from $0$.
This implies that for $r_L < \varrho < R_L$,
the convex boundary consists of $2$, $4$, or $6$ circular arcs
that alternate with the same number of circular arcs in the
concave boundary.
\begin{figure}[hbt]
  \centering \resizebox{!}{2.25in}{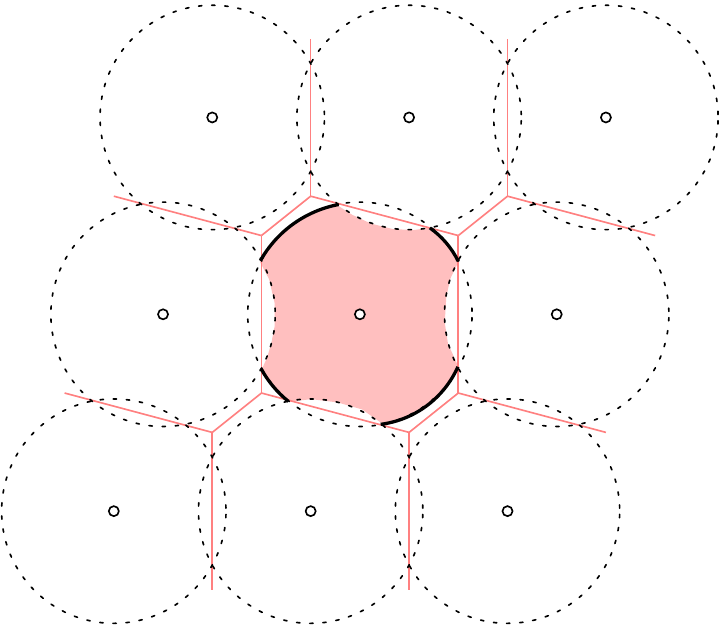}
  \caption{The shaded partial disk $D(0, \varrho)$ with its boundary
    divided into a solid convex portion and a dotted concave portion.}
 \label{fig:partialdisk}
\end{figure}

\paragraph{Angles.}
Recall that the convex boundary consists of at most three pairs of arcs,
and let $\varphi_i (\varrho)$ be the angle
of each of the two arcs in the $i$-th pair, for $i = 1,2,3$.
The total angle of the convex boundary is
\begin{align}
  \Phi_L (\varrho)  &=  \sum_{i=1}^3 2 \varphi_i (\varrho),
\end{align}
and the total angle of the concave boundary is
$2 \pi - \Phi_L (\varrho)$.
We have $\Phi_L (r_L) = 2 \pi$ and $\Phi_L (R_L) = 0$,
and between these two limits, the function is continuous and
monotonically decreasing.
\begin{lemma}[Monotonicity]
  Let $L$ be a lattice in $\Rspace^2$.
  Then $\Phi_L \colon [r_L, R_L] \to [0, 2 \pi]$ is continuous,
  with $\Phi_L (\varrho_1) > \Phi_L (\varrho_2)$ whenever $\varrho_1 < \varrho_2$.
\end{lemma}
\begin{proof}
  The continuity of the function follows from the fact that
  $\partial B(0, \varrho)$ intersects $\partial V(0)$
  in at most a finite number of points.

  To prove monotonicity, we recall that $\Phi_L (\varrho)$
  is the total angle of
  $\partial_x D(0,\varrho) = \partial B(0, \varrho) \cap V(0)$.
  The Voronoi domain is a convex polygon with $0 \in V(0)$.
  Drawing circles with radii $\varrho_1 < \varrho_2$ centered at $0$,
  we let $0 \leq \theta < 2 \pi$ and write
  $p_1 (\theta)$ and $p_2 (\theta)$ for the points on the circles
  in direction $\theta$.
  Either both points belong to $V(0)$, both points do not belong to $V(0)$,
  or $p_1 (\theta) \in V(0)$ but $p_2 (\theta) \not\in V(0)$.
  The fourth combination is not possible,
  which implies $\Phi_L (\varrho_1) \geq \Phi_L (\varrho_2)$.
  To prove the strict inequality, we just need to observe
  that there is an arc of non-zero length in $\partial_x D(0, \varrho_1)$
  such that the corresponding arc in $\partial B(0, \varrho_2)$
  lies outside $V(0)$ and therefore does not belong to
  $\partial_x D(0, \varrho_2)$.
\end{proof}

\paragraph{Area.}
The probability that a random point belongs to exactly one disk
is the area of $D(0, \varrho)$ over the area of $V(0)$.
The latter is a constant independent of the radius.
We will prove shortly that the former is a unimodal function
in $\varrho$ with a single maximum at the radius $\varrho = \varrho_L$
that balances the lengths of the two kinds of boundaries;
see Figure \ref{fig:equilibrium}.
We call $\varrho_L$ the \emph{equilibrium radius} of $L$.
Write $A_L \colon [r_L, R_L] \to \Rspace$ for the function that
maps $\varrho$ to the area of $D(0, \varrho)$.
\begin{figure}[hbt]
  \centering \resizebox{!}{2.0in}{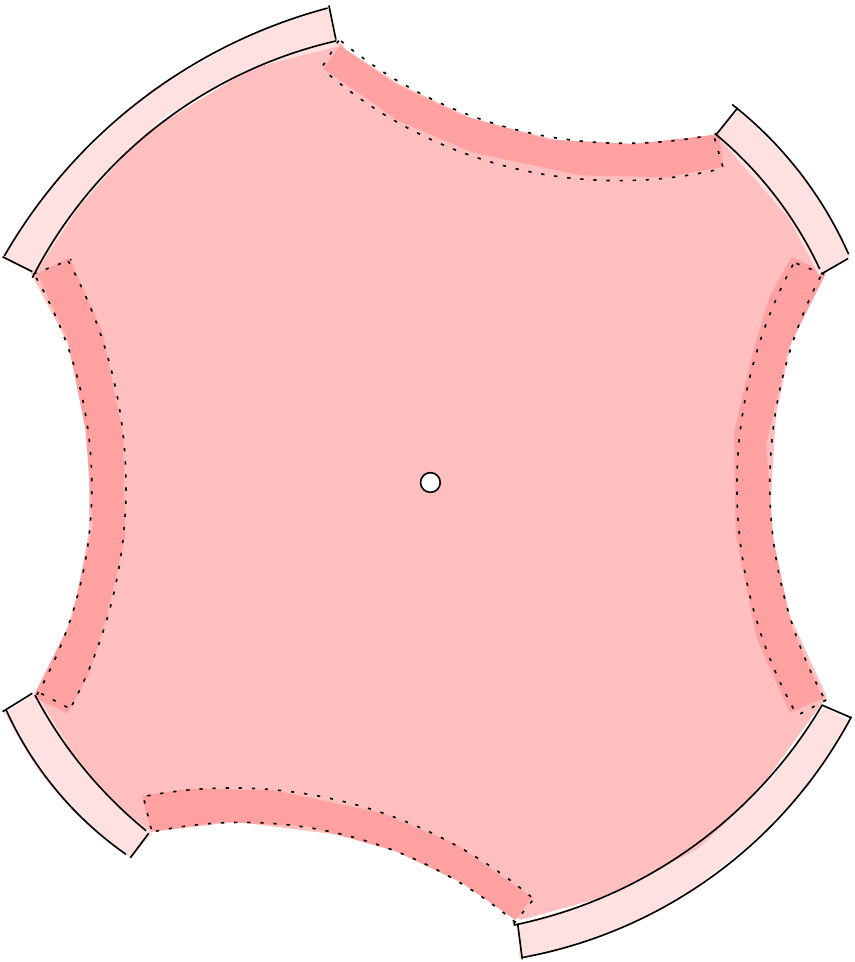}
  \caption{Increasing the radius grows $D(0, \varrho)$ along
    the convex boundary and shrinks it along the concave boundary.}
 \label{fig:equilibrium}
\end{figure}
\begin{lemma}[Equilibrium Radius]
  Let $L$ be a lattice in $\Rspace^2$.
  The function $A_L \colon [r_L, R_L] \to \Rspace$ is strictly concave,
  with a unique maximum at the equilibrium radius $\varrho_L$ that satisfies
  $\Phi_L (\varrho_L) = \pi$.
\end{lemma}
\begin{proof}
  Recall that $\Phi_L (\varrho)$ is the total angle of the convex boundary
  of $D(0, \varrho)$, and $2 \pi - \Phi_L (\varrho)$ is the total
  angle of the concave boundary.
  When we increase the radius, the partial disk grows along the convex
  boundary and shrinks along the concave boundary.
  Indeed, the derivative is the difference between the two lengths:
  \begin{align}
    \tfrac{\partial A_L}{\partial \varrho} (\varrho)
      &=  \varrho [\Phi_L (\varrho) - (2 \pi - \Phi_L (\varrho)] \\
      &=  2 \varrho [\Phi_L (\varrho) - \pi].
  \end{align}
  The derivative vanishes when $\Phi_L (\varrho) = \pi$, is positive when
  $\Phi_L (\varrho) > \pi$ and negative when $\Phi_L (\varrho) < \pi$,
  as claimed.
\end{proof}

\section{Optimality of the Regular Hexagonal Grid}
\label{sec4}

In this section, we present the proof of our main result.
After writing the probability that a random point lies in exactly one
disk as a function of the radius, we distinguish between three cases,
showing that the maximum is attained at the regular hexagonal grid.

\paragraph{Probability.}
Given a lattice $L$ in $\Rspace^2$, we write $r_L < \varrho_L < R_L$
for the packing, equilibrium, and covering radii.
Recall that the probability in question is
\begin{align}
  P_L (\varrho_L)  &=  \tfrac{A_L (\varrho_L)}{\norm{a} \norm{b} \sin \gamma} ,
\end{align}
in which $L = L(a,b)$ and $\gamma$ is the angle between $a$ and $b$.
Recall furthermore that the convex boundary consists of at most six arcs,
two each with angle $\varphi_1$, $\varphi_2$, $\varphi_3$,
in which we set the angle to zero if the arc degenerates to a point
or is empty.
\begin{lemma}[Equilibrium Area]
  Let $L = L(a,b)$ be a lattice in $\Rspace^2$
  with angle $\gamma$ between $a$ and $b$.
  Then
  \begin{align}
    P_L (\varrho_L)  &=  \tfrac{2 \varrho_L^2}{\norm{a} \norm{b} \sin \gamma}
                        \cdot \sum_{i=1}^3 \sin \varphi_i  .
    \label{eqn:probability}
  \end{align}
\end{lemma}
\begin{proof}
  Recall that $\norm{a} \norm{b} \sin \gamma$ is the area of $V(0)$.
  Let $A_{\rm in}$ be the area of $B(0, \varrho_L) \cap V(0)$,
  let $A_{\rm out}$ be the area of $B(0, \varrho_L) \setminus V(0)$,
  and note that $A_{\rm in} - A_{\rm out}$ is the area of $D(0, \varrho_L)$.
  Since $A_{\rm in} + A_{\rm out} = \varrho_L^2 \pi$,
  we have $A_{\rm in} - A_{\rm out} = \varrho_L^2 \pi - 2 A_{\rm out}$.
  The portion of $B(0, \varrho_L)$ outside the Voronoi domain consists
  of up to three symmetric pairs of disk segments, with total area
  \begin{align}
    A_{\rm out}  &=  2 \sum_{i=1}^3 \tfrac{\varrho_L^2}{2}
                     (\varphi_i - \sin \varphi_i)                 \\
                 &=  \varrho_L^2 \left( \tfrac{\pi}{2}
                       - \sum_{i=1}^3 \sin \varphi_i \right) ,
    \label{eqn:areaout}
  \end{align}
  in which the second line is obtained using
  $\sum_{i=1}^3 \varphi_i = \tfrac{\pi}{2}$ from the Equilibrium Radius Lemma.
  The probability is $A_{\rm in} - A_{\rm out}$ divided by the area
  of the Voronoi domain:
  \begin{align}
    P_L (\varrho_L)  &=  \tfrac{\varrho_L^2 \pi - 2 A_{\rm out}}
                              {\norm{a} \norm{b} \sin \gamma}.
  \end{align}
  Together with \eqref{eqn:areaout} this implies the claimed relation.
\end{proof}

\paragraph{Case analysis.}
We focus on the primitive case in which the Voronoi domain is a hexagon,
considering the non-primitive case a limit situation in which two
of the edges shrink to zero length.
Let $a, b \in \Rspace^2$ be generators of the lattice satisfying
the condition in the Non-obtuse Generators Lemma,
set $c = a-b$, and assume $\norm{a} \leq \norm{b} \leq \norm{c}$.
Recall that these three lengths are the distances between
parallel edges of the hexagon.
Further notice that we have $r_L = \tfrac{\norm{a}}{2}$ for the packing radius
and $R_L > \tfrac{\norm{c}}{2}$ for the covering radius.
As before, we write $\varphi_i$ for the angles of the arcs
of $\partial_x D(0, \varrho)$,
and we index such that $\varphi_1 \geq \varphi_2 \geq \varphi_3$.
\begin{description}\denselist
  \item[{\sc Case 1:}]  $\tfrac{\norm{a}}{2} < \varrho_L \leq \tfrac{\norm{b}}{2}$.
    Then $\varphi_1 > 0$ and $\varphi_2 = \varphi_3 = 0$.
  \item[{\sc Case 2:}]  $\tfrac{\norm{b}}{2} < \varrho_L \leq \tfrac{\norm{c}}{2}$.
    Then $\varphi_1 \geq \varphi_2 > 0$ and $\varphi_3 = 0$.
  \item[{\sc Case 3:}]  $\tfrac{\norm{c}}{2} < \varrho_L < R_L$.
    Then $\varphi_1 \geq \varphi_2 \geq \varphi_3 > 0$.
\end{description}
For example the configuration depicted in Figure \ref{fig:partialdisk}
falls into Case 2.
Using the expression for the probability in the Equilibrium Area Lemma,
we determine the maximum for each of the three cases.
Here we state the results,
referring to Appendix \ref{appA} for the proofs.
By \emph{the probability} we mean of course the probability that
a random point belongs to exactly one disk.
\begin{lemma}[Two Arcs]
  In Case 1, the maximum probability is attained
  for $\norm{b} = \sqrt{2} \norm{a}$,  
  $\gamma = \arccos \tfrac{1}{2 \sqrt{2}}$,
  and $\varrho_L = \norm{a}/\sqrt{2}$,
  which gives $P_L (\varrho_L) = 0.755\ldots$.
\end{lemma}
\begin{lemma}[Four Arcs]
  In Case 2, the maximum probability is attained for
  $\norm{b} = \norm{a}$, $\gamma = \arccos (\sqrt{2} - 1)$,
  and $\varrho_L = \norm{a}\sqrt{1 - 1/\sqrt{2}}$,
  which gives  $P_L (\varrho_L) = 0.910\ldots$.
\end{lemma}
\begin{lemma}[Six Arcs]
  In Case 3, the maximum probability is attained for
  $\norm{b} = \norm{a}$ = \norm{c},
  $\gamma = \tfrac{\pi}{3}$ and $\varrho_L = \norm{a}/(2\cos\tfrac{\pi}{12})$,
  which gives  $P_L (\varrho_L) = 0.928\ldots$.
\end{lemma}
Note that the lattice in the Six Arcs Lemma is the regular hexagonal grid.
Comparing the three maximum probabilities,
we see that the regular hexagonal grid gives the global optimum;
see Figures \ref{fig:prob}.
\begin{figure}[hbt]
  \centering
    {\includegraphics[width=3.7cm]{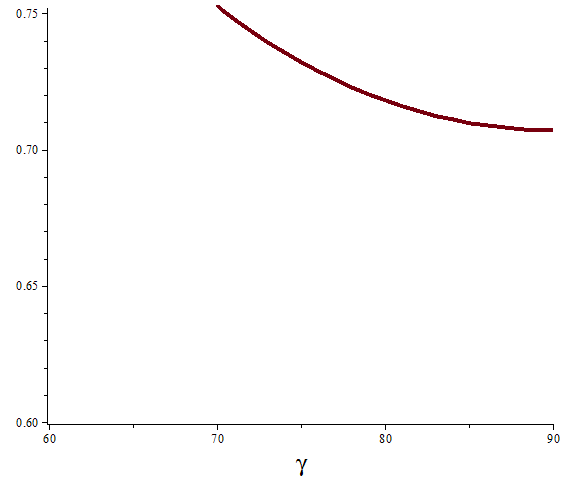}}
    \hspace{0.2in}
    {\includegraphics[width=3.7cm]{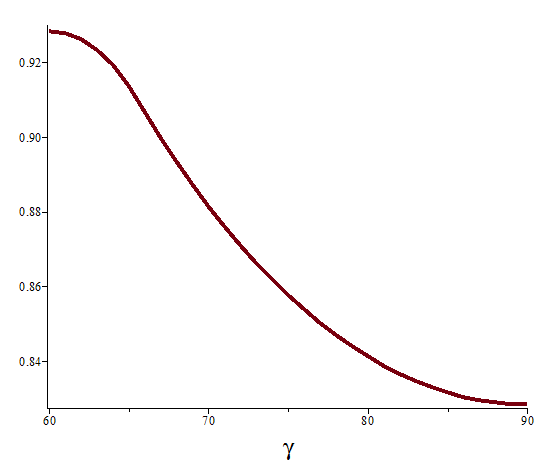}}
  \caption{\emph{Left:} for $\sqrt{2} \norm{a} = \norm{b} \leq \norm{c}$,
    the angle $\gamma$
    is between $\arccos \tfrac{1}{2 \sqrt{2}}$ and $\tfrac{\pi}{2}$.
    \emph{Right:} for $\norm{a} = \norm{b} \leq \norm{c}$,
    the angle $\gamma$ is between $\tfrac{\pi}{3}$ and $\tfrac{\pi}{2}$.
    In both cases, the probability increases as the angle decreases,
    attaining its maximum at the minimum angle.}
 \label{fig:prob}
\end{figure}
For this lattice, we get $\varrho_H$ such that each disk
overlaps with six others and in each case covers $30^\circ$
of the bounding circle: the $12$-hour clock configuration in the plane.
This implies the Main Theorem stated in Section \ref{sec1}.
We further illustrate the result by showing the graph of the function
that maps $\norm{a} / \norm{b}$ and the angle $\gamma$ to the probability
at the equilibrium radius; see Figure \ref{fig:capregions}.
\begin{figure}[hbt]
  \centering
  \includegraphics[width=12cm]{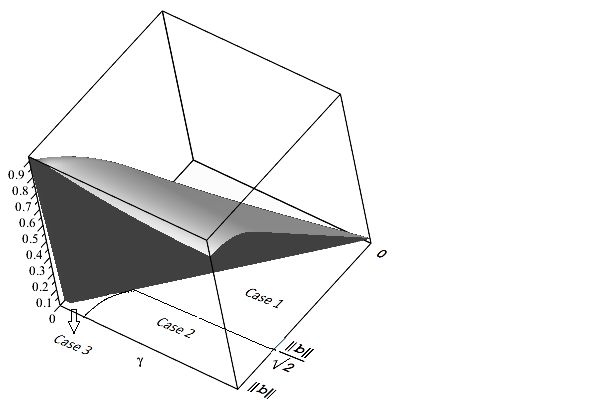}
  \caption{The graph of the function that maps
    $0 \leq \norm{a}/\norm{b} \leq 1$ and
    $\arccos \tfrac{\norm{a}}{2 \norm{b}} \leq \gamma \leq 90^\circ$
    to the probability that a random point lies in exactly one disk.
    The thus defined domain resembles a triangle and
    decomposes into three regions corresponding to Cases 1, 2, 3.
    The regular hexagonal grid is located at the lower left corner
    of the domain.}
  \label{fig:capregions}
\end{figure}

\section{Discussion}
\label{sec5}

The main result of this paper is a proof that the $12$-hour clock
configuration of disks in the plane maximizes the probability
that a random point lies in exactly one of the disks.
Other criteria favoring configurations between packing and covering
can be formulated, see \cite{Ham14},
and it would be interesting to decide which one fits the
biological data about DNA organization within the nucleus best.
There are also concrete mathematical questions
related to the work in this paper:
\begin{itemize}\denselist
  \item  Is the $12$-hour clock configuration optimal among
    all configurations of congruent disks in the plane?
  \item  What is the optimal lattice configuration of balls
    in $\Rspace^3$?
\end{itemize}
To appreciate the difficulty of the second question, we note that
the FCC lattice gives the densest packing \cite{Gauss63},
while the BCC lattice gives the sparsest covering \cite{Bambah54}.
Does one of them also maximize the probability that a random point
lies inside exactly one ball?

\section{Acknowledgements}
Vitaliy Kurlin is grateful to Herbert Edelsbrunner for hosting a short visit at IST Austria in March 2015, which was partially funded by the ACAT network.
The authors are grateful for the collaboration of Michael Kerber on proof checking the maple file that supports the computations on this paper.

\newpage

\appendix

\newpage
\section{Proofs}
\label{appA}

In this appendix, we give detailed proofs of the three Arc Lemmas.
As described in Section \ref{sec4}, the three lemmas add up to
a proof of the Main Theorem stated in Section \ref{sec1} of this paper.
We begin with a few relations that will be useful in all three proofs.
Given a triangle with edges of lengths $\norm{a}, \norm{b}, \norm{c}$
and angle $\gamma$ opposite the edge $c$, the \emph{law of cosines}
implies
\begin{align}
  \norm{c}^2  &=  \norm{a}^2 + \norm{b}^2 - 2 \norm{a} \norm{b} \cos \gamma .
  \label{eqn:lawofcosines}
\end{align}
Assuming $\norm{a} \leq \norm{b} \leq \norm{c}$, the angle $\gamma$
is at least as large as each of the other two angles.
From \eqref{eqn:lawofcosines} together with $\norm{b} \leq \norm{c}$,
we get $\cos \gamma \leq \tfrac{\norm{a}}{2 \norm{b}}$.
As justified by the Non-obtuse Generators Lemma, we may assume
the triangle is non-obtuse, which implies
\begin{align}
  \arccos \tfrac{\norm{a}}{2 \norm{b}}  &\leq  \gamma  \leq  \tfrac{\pi}{2} .
\end{align}
The range of possible angles is largest for $\norm{a} = \norm{b}$
where we get $60^\circ \leq \gamma \leq 90^\circ$.
Furthermore, we write the angles $\varphi_i$ of the arcs
in the convex boundary of the partial disk in terms of the edge lengths
and the radius:
\begin{align}
  \cos \tfrac{\varphi_1}{2}  &=  \tfrac{\norm{a}}{2 \varrho} ,
    \label{eqn:varphi1}                                        \\
  \cos \tfrac{\varphi_2}{2}  &=  \tfrac{\norm{b}}{2 \varrho} , 
    \label{eqn:varphi2}                                        \\
  \cos \tfrac{\varphi_3}{2}  &=  \tfrac{\norm{c}}{2 \varrho} .
    \label{eqn:varphi3}
\end{align}
The first relation holds provided
$\tfrac{\norm{a}}{2} \leq \varrho < R_L$,
and similar for the second and third relations.
Finally, we note that scaling does not affect the density of a configuration.
We can therefore set $\norm{b} = 1$, which we will do to simplify computations.

\paragraph{Proof of the Two Arcs Lemma.}
Case 1 is defined by $\tfrac{\norm{a}}{2} < \varrho_L \leq \tfrac{\norm{b}}{2}$,
which implies $\varphi_1 > 0$ and $\varphi_2 = \varphi_3 = 0$.
Since $\sin \varphi_2 = \sin \varphi_3 = 0$,
the probability at the equilibrium radius simplifies to
\begin{align}
  P_L (\varrho_L)  &=  \tfrac{2 \varrho_L^2}{\norm{a} \norm{b} \sin \gamma}
                       \cdot \sin \varphi_1 
    \label{eqn:TwoProb1} \\
                   &=  \tfrac{\norm{a}}{\norm{b} \sin \gamma} ,
    \label{eqn:TwoProb2}
\end{align}
where we get the second line by combining $\varphi_1 = \tfrac{\pi}{2}$
with \eqref{eqn:varphi1} to imply $\varrho_L = {\norm{a}} / {\sqrt{2}}$.
For the remainder of this proof, we normalize by setting $\norm{b} = 1$.
To maximize the probability, we choose $\norm{a}$ as large as possible
and $\gamma$ as small as possible.
From ${\norm{a}} / {\sqrt{2}} = \varrho_L \leq \tfrac{1}{2}$,
we get $\norm{a} \leq 1 / \sqrt{2}$,
and from $\norm{b} \leq \norm{c}$
we get $\gamma \geq \arccos \tfrac{1}{2 \sqrt{2}}$.
The two parameters can be optimized simultaneously, which gives
\begin{align}
  P_L (\varrho_L)  &=  \tfrac{1}{\sqrt{2}
                       \sin \left( \arccos \tfrac{1}{2 \sqrt{2}} \right)}
                    =  0.755\ldots .
\end{align}

\paragraph{Proof of the Four Arcs Lemma.}
Case 2 is defined by $\tfrac{\norm{b}}{2} < \varrho_L \leq \tfrac{\norm{c}}{2}$,
which implies $\varphi_1 \geq \varphi_2 > 0$ and $\varphi_3 = 0$.
The probability at the equilibrium radius is therefore
\begin{align}
  P_L (\varrho_L)  &=  \tfrac{2 \varrho_L^2}{\norm{a} \norm{b} \sin \gamma}
                       \cdot ( \sin \varphi_1 + \sin \varphi_2 ) .
  \label{eqn:FourProb}
\end{align}
To get a handle on the maximum of this function,
we first write the the sum of $\sin \varphi_1$ and $\sin \varphi_2$ 
and second the equilibrium radius in terms of other parameters.
Using $\cos 2 \alpha = \cos^2 \alpha - \sin^2 \alpha$ and 
 $\cos^2 \alpha + \sin^2 \alpha = 1$,
we get $\cos \varphi_2 = 2 \cos^2 \tfrac{\varphi_2}{2} - 1$, and
since $\varphi_1 + \varphi_2 = \frac{\pi}{2}$, we have $\sin \varphi_1 = \cos \varphi_2$.
Recalling \eqref{eqn:varphi2}, we get
$\sin \varphi_1 = {\norm{b}^2} / {(2 \varrho_L^2)} - 1$,
and recalling \eqref{eqn:varphi1}, we get
$\sin \varphi_2 = {\norm{a}^2} / {(2 \varrho_L^2)} - 1$.
Adding the two relations gives
\begin{align}
  \sin \varphi_1 + \sin \varphi_2 
    &=  \tfrac{\norm{a}^2 + \norm{b}^2}{2 \varrho_L^2} - 2 .
  \label{eqn:sinplussin}
\end{align}
To find a substitution for the equilibrium radius, we begin with
\eqref{eqn:varphi2},
use $\tfrac{\varphi_2}{2} = \tfrac{\pi}{4} - \tfrac{\varphi_1}{2}$,
and finally apply
$\cos (\alpha+\beta) = \cos \alpha \cos \beta - \sin \alpha \sin \beta$:
\begin{align}
  \tfrac{\norm{b}}{2 \varrho_L}
    &=  \cos \left( \tfrac{\pi}{4} - \tfrac{\varphi_1}{2} \right)  \\
    &=  \tfrac{1}{\sqrt{2}} \left( \cos \tfrac{\varphi_1}{2}
                                 + \sin \tfrac{\varphi_1}{2} \right) .
\end{align}
Next, we substitute the two trigonometric functions
using \eqref{eqn:varphi1} and $\sin^2 \alpha = 1 - \cos^2 \alpha$.
Simplifying the resulting relation and squaring it, we get
\begin{align}
  \varrho_L^2  &=  \tfrac{1}{2}
    \left( \norm{a}^2 + \norm{b}^2 - \sqrt{2} \norm{a} \norm{b} \right) .
  \label{eqn:fourarcsradius}
\end{align}
Plugging \eqref{eqn:sinplussin} and \eqref{eqn:fourarcsradius} into the
equation for the probability and normalizing by setting $\norm{b} = 1$,
we get
\begin{align}
  P_L (\varrho_L)  &=  \tfrac{2 \sqrt{2} \norm{a} - \norm{a}^2 - 1}
                             {\norm{a} \sin \gamma}                  \\
                   &=  \tfrac{2 \sqrt{2} \norm{a} - \norm{a}^2 - 1}
    {{\norm{a}}\sqrt{1 -
      \left( \sqrt{2}-\tfrac{\norm{a}^2+1}{2 \norm{a}} \right)^2 }}
    \label{eqn:fourarcsprob}
\end{align}
where we maximize to get the second line by
choosing $\gamma$ as small as possible.
Specifically, $\gamma$ is implicitly restricted
by $\tfrac{\norm{b}}{2} < \varrho_L \leq \tfrac{\norm{c}}{2}$, so we can
use $4 \varrho_L^2 \leq \norm{c}^2$ together with
\eqref{eqn:lawofcosines} and \eqref{eqn:fourarcsradius} to get
\begin{align}
  \cos \gamma  &\leq  \sqrt{2} - \tfrac{\norm{a}^2 + \norm{b}^2}
                                       {2 \norm{a} \norm{b}} .
  \label{eqn:fourarcsgamma}
\end{align}
Checking with the Maple software \cite{maple},
we find that the right-hand-side of \eqref{eqn:fourarcsprob}
increases in $[0,1]$ attaining its maximum at $\norm{a} = 1$.
We therefore get $\gamma = \arccos \left( \sqrt{2} - 1 \right)$
from \eqref{eqn:fourarcsgamma},
$\varrho_L^2 = 1 - 1/\sqrt{2}$ from \eqref{eqn:fourarcsradius}, and
\begin{align}
  P_L (\varrho_L)  &=  \sqrt{ 2 \sqrt{2} - 2 } 
                    =  0.910\ldots 
\end{align}
from \eqref{eqn:fourarcsprob}.

\paragraph{Proof of the Six Arcs Lemma.}
Case 3 is defined by $\tfrac{\norm{c}}{2} < \varrho_L < R_L$,
which implies $\varphi_1 \geq \varphi_2 \geq \varphi_3 > 0$.
Starting with the expression for the probability given in the
Equilibrium Area Lemma, we first express the $\sin \varphi_i$ in
terms of the other parameters:
\begin{align}
  \sin \varphi_1  &=  \tfrac{\norm{a} \sqrt{4 \varrho_L^2 - \norm{a}^2}}
                            {2 \varrho_L^2} , 
    \label{eqn:sixarcsvarphi1}                                    \\
  \sin \varphi_2  &=  \tfrac{\norm{b} \sqrt{4 \varrho_L^2 - \norm{b}^2}}
                            {2 \varrho_L^2} , 
    \label{eqn:sixarcsvarphi2}                                    \\
  \sin \varphi_3  &=  1 -
    \tfrac{\left[ \norm{a} \sqrt{4 \varrho_L^2 - \norm{b}^2}
                + \norm{b} \sqrt{4 \varrho_L^2 - \norm{a}^2} \right]^2}
          {8 \varrho_L^4} .
    \label{eqn:sixarcsvarphi3}
\end{align}
To get \eqref{eqn:sixarcsvarphi1},
we use $\sin 2 \alpha = 2 \sin \alpha \cos \alpha$
with $\alpha = \tfrac{\varphi_1}{2}$, together with \eqref{eqn:varphi1}.
To get \eqref{eqn:sixarcsvarphi2},
we use the same trigonometric identity
with $\alpha = \tfrac{\varphi_2}{2}$, together with \eqref{eqn:varphi2}.
To get \eqref{eqn:sixarcsvarphi3},
we use the equilibrium condition together with
$\sin \varphi_3 = \sin (\tfrac{\pi}{2} - \varphi_1 - \varphi_2) = \cos (\varphi_1 + \varphi_2)
  = 1 - 2 \sin^2 \tfrac{\varphi_1 + \varphi_2}{2}$,
and finally substitute
$\sin (\alpha + \beta) = \sin \alpha \cos \beta + \sin \beta \cos \alpha$,
with $\alpha = \tfrac{\varphi_1}{2}$ and $\beta = \tfrac{\varphi_2}{2}$.
  
To do the same for $\sin \gamma$, we take the cosine of both sides
of the equilibrium condition, which is
$\tfrac{\varphi_3}{2} = \tfrac{\pi}{4}
                      - \tfrac{\varphi_1}{2} - \tfrac{\varphi_2}{2}$.
Writing $c_i = \cos \tfrac{\varphi_i}{2}$ and $s_i = \sin \tfrac{\varphi_i}{2}$,
for $i = 1,2,3$, and applying standard trigonometric identities, we get
\begin{align}
  c_3    &=  \tfrac{1}{\sqrt{2}} [ c_1 c_2 - s_1 s_2 + s_1 c_2 + c_1 s_2 ] . \\
  c_3^2  &=  \tfrac{1}{2} + (2 c_1 c_2^2 - c_1) \sqrt{1-c_1^2} \nonumber \\
           &~~~~~~~~  + (2 c_1^2 c_2 - c_2) \sqrt{1-c_2^2} .
\end{align}
Using \eqref{eqn:varphi1}, \eqref{eqn:varphi2}, \eqref{eqn:varphi3} and 
substituting $\norm{c}^2$ using \eqref{eqn:lawofcosines},
we get the following relation after a few rearrangements:
\begin{align}
  \cos \gamma
   &=  \tfrac{\norm{a}^2 + \norm{b}^2 - 2}{2 \norm{a} \norm{b}} 
     - \tfrac{\norm{b}^2-2\varrho_L^2}{4 \varrho_L^2 \norm{b}} \sqrt{4\varrho_L^2-\norm{a}^2}
       \nonumber \\
   &   ~~~~~~~~~~~~~~~~~~~~~~~~~
     - \tfrac{\norm{a}^2-2\varrho_L^2}{4 \varrho_L^2 \norm{a}} \sqrt{4\varrho_L^2-\norm{b}^2} .
  \label{eqn:squared}
\end{align}
Using $\cos^2 \gamma = 1 - \sin^2 \gamma$, we can substitute $\sin \gamma$
in the formula for $P_L (\varrho_L)$.
We thus arrived at a relation that gives the probability in terms of
$\norm{a}$, $\norm{b}$, and $\varrho_L$ only.
While being lengthy, this relation is readily obtained by plugging
\eqref{eqn:sixarcsvarphi1}, \eqref{eqn:sixarcsvarphi2}, \eqref{eqn:sixarcsvarphi3},
and \eqref{eqn:squared} into \eqref{eqn:probability}.
We therefore take the liberty to omit the formula here and refer
the interested reader to the website of the second author of this paper\footnote{A Maple
  file with the main steps in the formulas related with this paper is available at
  \url{http://mabelih9.wix.com/mabelhome#!publications/cee5}.}.

It remains to determine the parameters that maximize the probability.
To simplify this task, we normalize by setting $\norm{b} = 1$.
The probability is thus a function of two variables, $\norm{a}$ and $\varrho_L$.
Using the Maple software, we compute the two partial derivatives,
$\partial P_L / \partial \norm{a}$ and $\partial P_L / \partial \varrho_L$.
Setting both to zero, we get $\norm{a} = 1$ matched up with
three radius values:
\begin{align}
  \varrho_L  &=  \tfrac{1}{2} \left( \sqrt{6}-\sqrt{2} \right) = 0.517\ldots, 
    \label{eqn:firstroot}          \\
  \varrho_L  &=  \tfrac{1}{2} \sqrt{C^\frac{1}{3}-1+C^{-\frac{1}{3}}} =  0.582\ldots,
    \label{eqn:secondroot}         \\
  \varrho_L  &=  \tfrac{1}{2} \left( \sqrt{6}+\sqrt{2} \right) = 1.931\ldots,
    \label{eqn:thirdroot}
\end{align}
with $C = 3 + 2 \sqrt{2}$.
Setting $\norm{a} = \norm{b} = 1$, the covering radius depends only on the angle $\gamma$,
which ranges from $\tfrac{\pi}{3}$ to $\tfrac{\pi}{2}$.
It is largest for $\gamma = \tfrac{\pi}{2}$, where $R_L = \sqrt{2}/2 = 0.707\ldots$.
The radius we get at the third root is larger than that and can therefore be excluded.

\begin{figure}[hbt]
  \centering
  \vspace{-0.1in} \includegraphics[width=9cm]{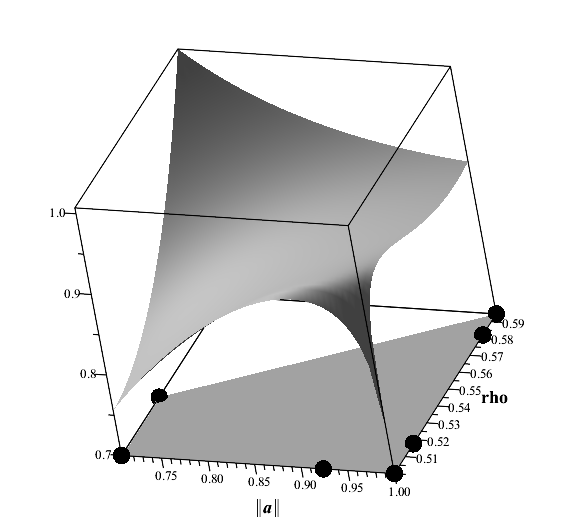} \vspace{-0.2in}
  \caption{The quadrangle in the plane defined by the length of $a$ and the radius
    is shaded.
    Along its boundary, we encounter three local minima (two corners and a point
    along the right edge) and three local maxima (a corner and a point each along
    the lower edge and the right edge).}
  \label{fig:superset}
\end{figure}
The maximum probability is attained at one of the two remaining roots or
along the boundary of the domain region that corresponds to Case 3.
As illustrated in Figure \ref{fig:superset}, we simplify the computation by taking
a quadrangle that contains this region.
The quadrangle is defined by
\begin{align}
  \tfrac{\sqrt{2}}{2}  &\leq  \norm{a}    \leq  1                         \\
  \tfrac{1}{2}         &\leq  \varrho~~~~ \leq  \tfrac{\norm{c}}{2\sin \gamma_{max}} ,    
\end{align}
in which the first interval follows from the bounds that define Case 3,
and the second interval is obtained by limiting the radius by the covering
radius of the configuration in which $a$ and $b$ enclose its maximum angle.
The maximum angle for Case 3, $\gamma_{max}$, corresponds to the 
minimal angle for Case 2 derived from~\eqref{eqn:fourarcsgamma}. 
The upper boundary on $\varrho$ is a convex curve and so we replace 
it with the straight line connecting its extremes.

We evaluate the probability at the four vertices and,
using the Maple software, at the roots of the derivatives along the four edges.
As shown in Figure \ref{fig:superset}, we find three local minima
alternating with three local maxima along the boundary of the quadrangle.
The maximum and minimum in the interior of the right edge coincide with the
roots in \eqref{eqn:firstroot} and \eqref{eqn:secondroot}.
Among the three local maxima, the probability is largest at \eqref{eqn:firstroot},
which is characterized by
$\norm{a} = \norm{b} = 1$ and
$\varrho_L = \tfrac{1}{2} \left( \sqrt{6} - \sqrt{2} \right) = 1 / (2 \cos \tfrac{\pi}{12})$.
This gives $P_L (\varrho_L) = 0.928\ldots$, as claimed in the Six Arcs Lemma.
Indeed, plugging the values into the equilibrium condition of Case 3,
we get $\norm{c} = 1$, which shows that the probability is maximized
by the regular hexagonal grid.

\end{document}